\documentclass[nofootinbib,
reprint,prd,nopacs,
nobibnotes,
amsmath,amssymb,
aps,
]{revtex4-1}

\usepackage{graphicx}
\usepackage{dcolumn}
\usepackage{bm}

\usepackage{amsmath}
\usepackage{graphicx}
\usepackage{amsfonts}
\usepackage{latexsym}
\usepackage{bbold}
\usepackage{wasysym}
\usepackage{calligra}
\usepackage{float}
\usepackage{ulem}
\usepackage{inputenc}
\usepackage{xspace}
\usepackage{url}
\usepackage{epstopdf}
\usepackage{tikz}
\usepackage{amsthm}

\usepackage{inputenc}

\newtheorem{mydef}{Definition}
\newtheorem{cor}{Corollary}

\newtheorem{prop}{Proposition}
\newtheorem{theorem}{Theorem}

\newcommand{\be}{\begin{equation}}
\newcommand{\ee}{\end{equation}}
\newcommand{\beq} {\begin{equation}}
\newcommand{\eeq} {\end{equation}}
\newcommand{\ba}{\begin{eqnarray}}
\newcommand{\ea}{\end{eqnarray}}

\begin{document}

	\title{Linear Transformations on Affine-Connections}

\author{Damianos Iosifidis}
\affiliation{Institute of Theoretical Physics, Department of Physics
	Aristotle University of Thessaloniki, 54124 Thessaloniki, Greece}
\email{diosifid@auth.gr}

	\begin{abstract}
	
	We state and prove a simple Theorem that allows one to generate invariant quantities in Metric-Affine Geometry, under a given transformation of the affine connection. We start by a general functional of the metric and the connection and consider transformations of the affine connection possessing a certain symmetry. We show that the initial functional is invariant under the aforementioned group of transformations iff its $\Gamma$-variation produces tensor of a given symmetry. Conversely if the tensor produced by the  $\Gamma$-variation of the functional respects a certain symmetry then the functional is invariant under the associated transformation of the affine connection. We then apply our results in Metric-Affine Gravity and produce invariant actions under certain transformations of the affine connection. Finally, we derive the constraints put on the hypermomentum for such invariant Theories.
	
\end{abstract}

\maketitle

\allowdisplaybreaks



\section{Introduction}
\label{intro}

The geometric structure of a  manifold, of a generic dimension $n$, is specified once a metric $g$ and an affine connection $\nabla$ are given. Then, the manifold is denoted as ($\mathcal{M},g,\nabla$) and in general possesses curvature, torsion and non-metricity. Setting both torsion and non-metricity to zero, one recovers the so-called Riemannian Geometry\cite{eisenhart1997riemannian} which is completely characterized by the metric alone \footnote{As it is well known, in this case the connection is the Levi-Civita and it is completely  determined by the metric tensor and its derivatives.}. It is also possible to impose vanishing curvature and non-metricity and obtain the well known teleparallel scheme sourced by torsion\cite{aldrovandi2012teleparallel}. Yet another possibility is to impose vanishing curvature and torsion and allow only for non-metricity to arrive to a symmetric teleparallel geometry\cite{nester1998symmetric,jimenez2018teleparallel}. More generally, imposing only vanishing curvature one obtains a generalized teleparallel geometry admitting both torsion and non-metricity\cite{jimenez2019general}.
 Letting all the above geometrical properties of the manifold unconstrained  we are in the realm of a non-Riemannian Geometry\cite{eisenhart2012non}. The corresponding Gravity Theory that is described by such a generalized geometry is called Metric-Affine Gravity(MAG)\cite{hehl1995metric,hehl1999metric,iosifidis2019metric,iosifidis2019exactly} and admits both torsion and non-metricity along with curvature.

Metric-Affine Gravity provides a well motivated and important generalization of General Relativity, where the intrinsic properties of matter like spin, shear and dilation contribute to the Gravitational field through the hypermomentum tensor\cite{hehl1976hypermomentum}. In the same manner that the energy momentum tensor induces curvature, the hypermomentum induces  torsion and non-metricity which setup a non-Riemannian arena for Gravity. Interesting models of matter with hypermomentum (hyperfluids) have been proposed in \cite{obukhov1993hyperfluid,obukhov1996model} and their Cosmological implications are well worth investigating. We should note that if the Gravitational action is invariant under a certain transformation of the affine connection, this invariance gives restrictions on the hypermomentum. The most common example of the later being the projective invariance of the Einstein-Hilbert action \cite{hehl1995metric,vitagliano2011dynamics,bernal2017non,alfonso2017trivial,iosifidis2019exactly} which implies that the hypermomentum has an identically vanishing trace  when contracted in its first two indices. Therefore we see that connection transformations are very interesting and Gravitational actions that are invariant under certain groups of transformations of the affine connection will have associated matter actions that should also respect that symmetry. It is the purpose of this article to study such connection transformations and their relevance with functional invariance. More specifically, we state and prove a Theorem that relates the invariance of a functional, under certain transformations of the affine connection, with the symmetries of the tensor obtained by the $\Gamma$-variation of the functional.

The paper is organized as follows. We first define the basic geometrical objects of a non-Riemannian Geometry. Then we discuss general and special transformations of the affine connection and introduce the nomenclature and the definitions that we are going to be using throughout. We then state and prove our Theorem and derive three  corollaries that follow immediately. Finally, we discuss the applications of the Theorem in Metric-Affine Gravity and we show how it can be used to derive invariant Gravity actions under connection transformations. In particular we reproduce the results for the projective invariant quadratic (in torsion and non-metricity) action of \cite{iosifidis2019scale} and we also derive the constraints for an enhanced invariance. Considering a general class of MAG Theories we see how the invariance of the Gravity action under certain transformations of the connection, imposes constraints on the sources and more specifically on the form of the hypermomentum.  We then conclude our results and also discuss other possible applications.

\section{Geometrical Objects}
The structure of a manifold is determined by the metric tensor $g_{\mu\nu}$ which measures distances, defines dot products and raises and lowers indices and also by the affine connection $\Gamma^{\lambda}_{\;\;\;\mu\nu}$ (oftentimes denoted just by $\nabla$) which defines the parallel transfer of tensor fields through covariant differentiation. In our conventions the covariant derivative of, say, a $(1,1)$ type tensor reads
\beq
\nabla_{\alpha}T^{\mu}_{\;\;\nu}=\partial_{\alpha}T^{\mu}_{\;\;\nu}-\Gamma^{\lambda}_{\;\;\;\nu\alpha}T^{\mu}_{\;\;\lambda}+\Gamma^{\mu}_{\;\;\;\lambda\alpha}T^{\lambda}_{\;\;\nu}
\eeq
Continuing, we define the torsion tensor as the antisymmetric part of the affine connection
\beq
S_{\mu\nu}^{\;\;\;\lambda}:=\Gamma^{\lambda}_{\;\;\;[\mu\nu]}
\eeq
which naturally arises by acting the antisymmetrized covariant derivative on a scalar
\beq
\nabla_{[\mu}\nabla_{\nu]}\phi=S_{\mu\nu}^{\;\;\;\lambda}\nabla_{\lambda}\phi
\eeq
Letting the latter operator act on a vector $u^{\mu}$ we obtain
\begin{equation}
[\nabla_{\alpha} ,\nabla_{\beta}]u^{\mu}=2\nabla_{[\alpha} \nabla_{\beta]}u^{\mu}=R^{\mu}_{\;\;\;\nu\alpha\beta} u^{\nu}+2 S_{\alpha\beta}^{\;\;\;\;\;\nu}\nabla_{\nu}u^{\mu}
\end{equation}
where
\begin{equation}
R^{\mu}_{\;\;\;\nu\alpha\beta}:= 2\partial_{[\alpha}\Gamma^{\mu}_{\;\;\;|\nu|\beta]}+2\Gamma^{\mu}_{\;\;\;\rho[\alpha}\Gamma^{\rho}_{\;\;\;|\nu|\beta]}
\end{equation}
is the so-called Riemann or Curvature tensor and the horizontal bars around an index denote that this index is left out of the (anti)-symmetrization. In non-Riemannian Geometries the only symmetry of the curvature tensor is antisymmetry in its last two indices as it is obvious by its definition.  Without the use of any metric, we can construct the  two independent contractions
\beq
R_{\nu\beta}:=R^{\mu}_{\;\;\;\nu\mu\beta}\;,\; \hat{R}_{\alpha\beta}:=R^{\mu}_{\;\;\;\mu\alpha\beta}
\eeq
The former defines as usual the Ricci tensor while the latter is frequently refereed to as homothetic curvature and is of purely non-Riemannian origin. Once a metric is given\footnote{There are also the purely affine models where no notion of a metric is required\cite{eddington1923mathematical,eddington1921generalisation}. For applications of the later to inflation the reader may consult \cite{azri2017affine} and references therein.} we can form yet another  contraction
\beq
\check{R}^{\lambda}_{\;\;\kappa}:=R^{\lambda}_{\;\;\mu\nu\kappa}g^{\mu\nu}
\eeq
However, the Ricci scalar is still uniquely defined since
\beq
R:=R_{\mu\nu}g^{\mu\nu}=-\check{R}_{\mu\nu}g^{\mu\nu}\;,\;\; \hat{R}_{\mu\nu}g^{\mu\nu}=0
\eeq

Going back to the torsion tensor, by contraction in the last two indices we may define the torsion vector
\beq
S_{\mu}:=S_{\mu\lambda}^{\;\;\;\;\lambda}
\eeq
which exists for arbitrary space dimension-$n$. For $n=4$ in particular we can also define the torsion pseudo-vector
\beq
\tilde{S}_{\mu}:=\epsilon_{\mu\alpha\beta\gamma}S^{\alpha\beta\gamma}
\eeq
where  $\epsilon_{\mu\alpha\beta\gamma}$ is the $d-dimensional$ totally antisymmetric Levi-Civita tensor. Now, for generic geometries the metric tensor is not covariantly conserved, and exactly this incompatibility\footnote{If $\nabla_{\alpha}g_{\mu\nu}=0 $ for some connection $\Gamma^{\lambda}_{\;\;\;\mu\nu}$ the metric is said to be compatible with the connection.} defines the non-metricity tensor
\beq
Q_{\alpha\mu\nu}=-\nabla_{\alpha}g_{\mu\nu}
\eeq
We can then contract the above in two independent ways, to get the two non-metricity vectors
\beq
Q_{\alpha}:=Q_{\alpha\mu\nu}g^{\mu\nu}\;,\;\; q_{\nu}=Q_{\alpha\mu\nu}g^{\alpha\mu}
\eeq
The former goes by the name Weyl vector, while the latter does not seem to have a particular name in the literature. The general non-Riemannian space that has all of its geometrical objects unconstrained will be denoted as $L_{n}$. In what follows we consider connection transformations and expand on the tools we are going to use for the Theorem.

\section{Linear Transformations of Affine Connection}
A general shift transformation of the affine connection 
\beq
\Gamma^{\lambda}_{\;\;\;\mu\nu} \longrightarrow \Gamma^{\lambda}_{\;\;\;\mu\nu} \label{affcongen} +N^{\lambda}_{\;\;\mu\nu}
\eeq
where $N^{\lambda}_{\;\;\mu\nu}$ is a type $(1,2)$ tensor, induces changes on the various geometrical objects of an $L_{n}$. For instance, torsion and non-metricity change according to
\beq
S_{\mu\nu\alpha} \longrightarrow S_{\mu\nu\alpha}+N_{\alpha[\mu\nu]}
\eeq
\beq
Q_{\nu\alpha\mu}\longrightarrow Q_{\nu\alpha\mu}+2N_{(\alpha\mu)\nu}
\eeq
respectively. Two immediate results of the above are the following; any transformation of the form $(\ref{affcongen})$ with $N_{\alpha[\mu\nu]}=0$ preserves torsion and any transformation  of the form $(\ref{affcongen})$ with $N_{(\alpha\mu)\nu}=0$ preserves non-metricity. It is also obvious there there is no transformation of the affine connection that preserves both torsion and non-metricity at the same time. In addition, under the connection transformation $(\ref{affcongen})$ the Riemann tensor transforms as
\begin{gather}
R^{\mu}_{\;\;\nu\alpha\beta} \longrightarrow R^{\mu}_{\;\;\nu\alpha\beta} +2\nabla_{[\alpha}N^{\mu}_{\;\;\;|\nu|\beta]} \\ \nonumber
-2S_{\alpha\beta}^{\;\;\;\;\lambda}N^{\mu}_{\;\;\nu\lambda}+2 N^{\mu}_{\;\;\;\lambda [\alpha}N^{\lambda}_{\;\;\;|\nu|\beta]}
\end{gather}
where $\nabla_{\mu}$ and $S_{\alpha\beta}^{\;\;\;\;\lambda}$ are the covariant derivative and the torsion tensor computed with respect to the initial connection $\Gamma^{\lambda}_{\;\;\;\mu\nu}$.

 Let us now discuss vectorial transformations. In any dimension, there exist three independent such transformations\footnote{Here $\xi_{\mu}$, $\zeta_{\mu}$, and $\chi_{\mu}$ are arbitrary one forms.} 
\beq
\Gamma^{\lambda}_{\;\;\mu\nu}\longrightarrow \Gamma^{\lambda}_{\;\;\mu\nu}+\delta^{\lambda}_{\mu}\xi_{\nu}
\eeq
\beq
\Gamma^{\lambda}_{\;\;\mu\nu}\longrightarrow \Gamma^{\lambda}_{\;\;\mu\nu}+\delta^{\lambda}_{\nu}\zeta_{\mu}
\eeq
\beq
\Gamma^{\lambda}_{\;\;\mu\nu}\longrightarrow \Gamma^{\lambda}_{\;\;\mu\nu}+\chi^{\lambda}g_{\mu\nu}
\eeq
which we shall call transformation of the \textit{$1^{st}$, $2^{nd}$ and $3^{rd}$ $kind$} respectively. This nomenclature will make sense in the proceeding sections. Transformations of the \textbf{$1^{st}$ kind} are also known as \textbf{projective transformations}\cite{schouten2013ricci,eisenhart2012non}. Note that in $4$-dimensions there is also another possibility, namely
\beq
\Gamma^{\lambda}_{\;\;\mu\nu}\longrightarrow \Gamma^{\lambda}_{\;\;\mu\nu}+\epsilon^{\lambda}_{\;\;\mu\nu\alpha}K^{\alpha}
\eeq
where $\epsilon_{\mu\nu\rho\sigma}$ is the Levi-Civita tensor and $K_{\mu}$ an arbitrary pseudo-vector. However, since we want to keep our discussion as general as possible (for any dimension) we will mainly consider transformations of the above three kinds. A \textbf{general} \textbf{vectorial} \textbf{transformation} of the affine connection will therefore read\footnote{It would be interesting to find the equivalent class of paths whose connections are related by this general vectorial transformation just like the equivalent class of projective connections \cite{eisenhart2012non,schouten2013ricci}(see also \cite{iosifidis2019torsion,bejarano2019geometric} ).}
\beq
\Gamma^{\lambda}_{\;\;\mu\nu}\longrightarrow \Gamma^{\lambda}_{\;\;\mu\nu}+\delta^{\lambda}_{\mu}\xi_{\nu}+\delta^{\lambda}_{\nu}\zeta_{\mu}+\chi^{\lambda}g_{\mu\nu} \label{vecgen}
\eeq
where $\xi_{\mu}$, $\zeta_{\mu}$, and $\chi_{\mu}$ are in general independent one forms. If these fields are parallel (i.e there exist $\lambda_{1}$,$\lambda_{2}$ and $\lambda_{3}$ $\in \mathbb{R}$ such that $\xi_{\mu}=\lambda_{1}v_{\mu},\;\zeta_{\mu}=\lambda_{2}v_{\mu}$, $\chi_{\mu}=\lambda_{3}v_{\mu}$ ) we have a \textbf{constrained vectorial transformation}
\beq
\Gamma^{\lambda}_{\;\;\mu\nu}\longrightarrow \Gamma^{\lambda}_{\;\;\mu\nu}+\lambda_{1}\delta^{\lambda}_{\mu}v_{\nu}+\lambda_{2}\delta^{\lambda}_{\nu}v_{\mu}+\lambda_{3}v^{\lambda}g_{\mu\nu} \label{veccons}
\eeq
Note that for specific values of the $\lambda_{i}'s$ one obtains the Weyl geometry and its torsionful extensions considered in \cite{jimenez2016spacetimes,beltran2017modified}.
We will come back to all of the above transformations and also discuss functional invariance with respect to those, after we present our Theorem. We shall now continue with some definitions.

\section{Definitions used for the Theorem}

Before presenting our Theorem, it would be first useful to introduce the two definitions below.
\begin{mydef}
	Let $N_{\alpha\mu\nu}$ be a type-($0,3$) tensor of a given symmetry $S_{N}$ in its indices. We shall write
	\beq
	N_{\alpha\mu\nu}=\hat{S}_{N}N_{\alpha\mu\nu}
	\eeq
	where $ \hat{S}_{N}$ is the symmetry operation applied on $N_{\alpha\mu\nu}$. This symmetry operation can also be applied to another tensor $\Psi_{\alpha\mu\nu}$ of the same rank and we write
	\beq
	\hat{S}_{N}\Psi_{\alpha\mu\nu}
	\eeq
	which says that the index symmetry of $N_{\alpha\mu\nu}$ is passed over to $\Psi_{\alpha\mu\nu}$.
\end{mydef}
The above definition shall become clear with the following example.
\newline
\textbf{Example.}
Let $N_{\alpha\mu\nu}$ be a tensor that is symmetric in its first two indices, we then have
\beq
N_{\alpha\mu\nu}=\hat{S}_{N}N_{\alpha\mu\nu}=N_{(\alpha\mu)\nu}
\eeq
where  in this instance $\hat{S}_{N}$ means  \textit{'Symmetrize in the first two indices'}. We also have
\beq
\hat{S}_{N}\Psi_{\alpha\mu\nu}=\Psi_{(\alpha\mu)\nu}
\eeq
as well as
\beq
N_{\alpha\mu\nu}\Psi^{\alpha\mu\nu}=(\hat{S}_{N}N_{\alpha\mu\nu})\Psi^{\alpha\mu\nu}=N_{(\alpha\mu)\nu}\Psi^{(\alpha\mu)\nu}
\eeq

\begin{mydef}
	Let $\Psi_{\lambda}^{\;\;\mu\nu}$ be a type $(2,1)$ tensor (or tensor density). Define the three independent contractions
	\beq
	\Psi_{(1)}^{\nu}:=\Psi_{\lambda}^{\;\;\lambda\nu}\;, \; \Psi_{(2)}^{\mu}:=\Psi_{\lambda}^{\;\;\mu\lambda}\;, \; \Psi_{(3)\lambda}:=\Psi_{\lambda}^{\;\;\mu\nu}g_{\mu\nu}
	\eeq
\end{mydef}
We shall call them $1^{st}$, $2^{nd}$ and $3^{rd}$ contraction of $\Psi_{\lambda}^{\;\;\mu\nu}$. That is the $1^{st}$ contraction involves contraction between first and second indices, the second between second and third etc. With the help of the above definitions we are now in a position to present our Theorem.

\section{The Theorem}

\begin{theorem}
	Let $\Psi(g,\Gamma)$ be a function of the metric $g_{\mu\nu}$ and the affine connection $\Gamma^{\lambda}_{\;\;\;\mu\nu}$ and let\footnote{In this context $\psi$ appears to be a scalar while $\Psi=\sqrt{|g|}\psi$ is a scalar density. For convenience we will work with $\Psi$ but of course it is obvious that identical results hold true for $\psi$ as well.}
	\beq
	J[g,\Gamma]=\int d^{n}x \sqrt{|g|}\psi(g,\Gamma) =\int d^{n}x \Psi(g,\Gamma) \label{theo}
	\eeq
	be a functional of these variables. Define the variational derivative
	\beq
	\Psi_{\lambda}^{\;\;\mu\nu}:=\frac{\delta \Psi}{\delta \Gamma^{\lambda}_{\;\;\;\mu\nu}}
	\eeq 
	If $\Psi$ (and subsequently $J$) is invariant under the transformation
	\beq
	\Gamma^{\lambda}_{\;\;\;\mu\nu}\longrightarrow \Gamma^{\lambda}_{\;\;\;\mu\nu}+N^{\lambda}_{\;\;\;\mu\nu}\label{gencon}
	\eeq
	where $N_{\alpha\mu\nu}=\hat{S}_{N}N_{\alpha\mu\nu}$ is a tensor of a given symmetry $\hat{S}_{N}$, then
	\beq
	\hat{S}_{N}\Psi^{\alpha\mu\nu}=0 \label{psiS}
	\eeq
	Conversely, if ($\ref{psiS}$) holds true the initial scalar $\Psi$ (and subsequently the functional $J$) is invariant under the connection transformation $(\ref{gencon})$.
\end{theorem}
\begin{proof}
	We will start by first proving the $''\Rightarrow''$ part. Since by hypothesis $\Psi$ (and also $J$) is invariant under the transformation $(\ref{gencon})$ provided that  $N_{\alpha\mu\nu}=\hat{S}_{N}N_{\alpha\mu\nu}$, we have 
	\beq
\delta_{\Gamma}J=0	
	\eeq
and by expanding we get
\begin{gather}
0=\int d^{n}x \frac{\delta \Psi}{\delta \Gamma^{\lambda}_{\;\;\;\mu\nu}}\delta \Gamma^{\lambda}_{\;\;\;\mu\nu}= \nonumber \\
=\int d^{n}x \Psi_{\lambda}^{\;\;\mu\nu}N^{\lambda}_{\;\;\mu\nu}=\int d^{n}x \Psi^{\alpha\mu\nu}(\hat{S}_{N}N_{\alpha\mu\nu}) = \nonumber \\
=\int d^{n}x (\hat{S}_{N}\Psi^{\alpha\mu\nu})(\hat{S}_{N}N_{\alpha\mu\nu})
\end{gather}
but since $S_{N}N_{\alpha\mu\nu}$ is arbitrary we conclude that
\beq
\hat{S}_{N}\Psi^{\alpha\mu\nu} =0 \label{psieq}
\eeq
Conversely, given that $(\ref{psieq})$ holds true we compute
\begin{gather}
J[g,\Gamma+N]-J[g,\Gamma]=\int d^{n}x \frac{\delta \Psi}{\delta \Gamma^{\lambda}_{\;\;\;\mu\nu}}\delta \Gamma^{\lambda}_{\;\;\;\mu\nu}= \nonumber \\
=\int d^{n}x \Psi_{\lambda}^{\;\;\mu\nu}N^{\lambda}_{\;\;\mu\nu}=\int d^{n}x \Psi^{\alpha\mu\nu}(\hat{S}_{N}N_{\alpha\mu\nu}) = \nonumber \\
=\int d^{n}x (\hat{S}_{N}\Psi^{\alpha\mu\nu})(\hat{S}_{N}N_{\alpha\mu\nu})=0
\end{gather}
That is
\beq
J[g,\Gamma+N]=J[g,\Gamma]
\eeq
and therefore the functional ($\ref{theo}$) is invariant under transformations of the form $(\ref{gencon})$ provided that $N_{\alpha\mu\nu}=\hat{S}_{N}N_{\alpha\mu\nu}$.

\end{proof}

As an immediate application of the above Theorem we claim the following corollaries.

\begin{cor}
	If the functional ($\ref{theo}$) is invariant under connection transformations of the $i-th$ $kind$ ($i=1,2,3$) then its  $\Gamma$-variation produces a tensor that has vanishing $i-th$ contraction
	\beq
	\Psi_{(i)}^{\mu}=0
	\eeq
	\end{cor}
and conversely, if the $i-th$ contraction of $\Psi_{\lambda}^{\;\;\mu\nu}$ vanishes identically, then the functional $(\ref{theo})$ is invariant under connection transformations of the $i-th$ $kind$.

\begin{cor}
If the functional ($\ref{theo}$) is invariant under general vectorial transformations of the connection of the form ($\ref{vecgen}$), then its $\Gamma$-variation produces a tensor that has all of its contractions vanishing, namely the tensor $\Psi_{\lambda}^{\;\;\mu\nu}$ is traceless in all of its contractions
	\beq
\Psi_{(i)}^{\mu}=0\;, \forall i=1,2,3
\eeq
Conversely, if the tensor obtained by the $\Gamma$-variation of ($\ref{theo}$) is traceless in all of its contractions then the action is invariant under general vectorial transformations of the connection of the type($\ref{vecgen}$).
\end{cor}

\begin{cor}
If the functional ($\ref{theo}$) is invariant under constrained vectorial transformations of the form ($\ref{veccons}$), then the traces of its $\Gamma$-variation $\Psi_{\lambda}^{\;\;\mu\nu}$ , satisfy
\beq
\sum_{i=1}^{3}\lambda_{i}\Psi_{(i)}^{\mu}=0
\eeq
That is, the three traces are  linearly dependent.
\end{cor}
Note that special cases of the above give the results of $Corollary$ $1$. In addition for the parameter choice $\lambda_{1}=\lambda_{2}=-\lambda_{3}$ one obtains invariance under a  Weyl transformation of the connection. To be more precise, if it holds that
\beq
\Psi_{(1)}^{\mu}+\Psi_{(2)}^{\mu}-\Psi_{(3)}^{\mu}=0
\eeq
then the functional ($\ref{theo}$) is invariant under connection transformations of the form
\beq
\Gamma^{\lambda}_{\;\;\mu\nu}\longrightarrow \Gamma^{\lambda}_{\;\;\mu\nu}+\lambda_{1}\Big( \delta^{\lambda}_{\mu}v_{\nu}+\delta^{\lambda}_{\nu}v_{\mu}-v^{\lambda}g_{\mu\nu} \Big)
\eeq
Let us now discuss some physical applications of the above Theorem.

\section{Applications to Gravity}
Our Theorem and its subsequent corollaries find its nature application on the scheme of Metric-Affine Gravity. For instance it can be applied for obtaining Theories that are invariant under projective transformations of the connection, or even more general ones. The powerful use of our Theorem is that it can be applied conversely, as we have shown.
 That is, we can obtain actions invariant under a specific group of connection transformations by simply looking at the symmetries and constraints of its $\Gamma$-variation.  The advantage is that even though it is a difficult task to find invariant actions, on the other hand, it is straightforward to take contractions and symmetrizations of a given tensor. Our result shows that one implies the other and vice versa and we therefore have a powerful tool for obtaining invariant actions just by looking at the symmetries of the corresponding $\Gamma$-variation of the action. The idea will become clear with the following examples.
 
\subsection{Invariant Theories}
\textbf{Example-Application.} Let us consider the quadratic action of \cite{iosifidis2019scale}
\begin{gather}
S
=\frac{1}{2 \kappa}\int d^{n}x \sqrt{-g} \Big[   
b_{1}S_{\alpha\mu\nu}S^{\alpha\mu\nu} +
b_{2}S_{\alpha\mu\nu}S^{\mu\nu\alpha} +
b_{3}S_{\mu}S^{\mu} \nonumber \\
a_{1}Q_{\alpha\mu\nu}Q^{\alpha\mu\nu} +
a_{2}Q_{\alpha\mu\nu}Q^{\mu\nu\alpha} +
a_{3}Q_{\mu}Q^{\mu}+
a_{4}q_{\mu}q^{\mu}+
a_{5}Q_{\mu}q^{\mu} \nonumber \\
+c_{1}Q_{\alpha\mu\nu}S^{\alpha\mu\nu}+
c_{2}Q_{\mu}S^{\mu} +
c_{3}q_{\mu}S^{\mu} \Big]  \label{genact}
\end{gather}
This is the most general parity even action that is quadratic in torsion and non-metricity \cite{iosifidis2019scale,jimenez2019general}. We will now show how one can restrict the parameter space in order to obtain a projective invariant Theory without computing the change in $S$ directly but by simply looking at its $\Gamma$-variation. Indeed, varying the above with respect to the connection, we have \cite{iosifidis2019scale}
\begin{gather}
\Psi_{\lambda}^{\;\;\mu\nu}= H^{\mu\nu}_{\;\;\;\;\lambda}+\delta^{\mu}_{\lambda}k^{\nu}+\delta^{\nu}_{\lambda}h^{\mu}+g^{\mu\nu}h_{\lambda}+f^{[\mu}\delta^{\nu ]}_{\lambda} 
\end{gather}
where
\begin{gather}
H^{\mu\nu}_{\;\;\;\;\lambda} = 4 a_{1}Q^{\nu\mu}_{\;\;\;\;\lambda}+2 a_{2}(Q^{\mu\nu}_{\;\;\;\;\lambda}+Q_{\lambda}^{\;\;\;\mu\nu})+2 b_{1}S^{\mu\nu}_{\;\;\;\;\lambda} \nonumber \\
+2 b_{2}S_{\lambda}^{\;\;\;[\mu\nu]}+c_{1}( S^{\nu\mu}_{\;\;\;\;\lambda}-S_{\lambda}^{\;\;\;\nu\mu}+Q^{[\mu\nu]}_{\;\;\;\;\;\lambda})
\end{gather}
\beq
k_{\mu} = 4 a_{3}Q_{\mu}+2 a_{5}q_{\mu}+2 c_{2}S_{\mu}
\eeq
\beq
h_{\mu} =a_{5} Q_{\mu}+2 a_{4}q_{\mu}+c_{3}S_{\mu}
\eeq
\beq
f_{\mu} = c_{2} Q_{\mu}+ c_{3}q_{\mu}+2 b_{3}S_{\mu}
\eeq
Taking the first contraction of $\Psi_{\lambda}^{\;\;\mu\nu}$ it follows that\footnote{Recall that the first contraction is defined by $\Psi_{(1)}^{\nu}:=\Psi_{\mu}^{\;\;\mu\nu}$}
\begin{gather}
\Psi_{(1)}^{\nu}=\Big[ 4a_{1}-\frac{c_{1}}{2}+4 n a_{3}+2 a_{5}+\frac{(1-n)}{2}c_{2} \Big] Q^{\nu} \nonumber \\
+\Big[ 4 a_{2}+\frac{c_{1}}{2}+2 n a_{5}+4 a_{4}+\frac{(1-n)}{2}c_{3} \Big] \tilde{Q}^{\nu} \nonumber \\
+\Big[ -2b_{1}+b_{2}+2 c_{1}+2 n c_{2}+2 c_{3}+(1-n)b_{3} \Big] S^{\nu}
\end{gather}
Now demanding a vanishing first contraction we get the constraints
\beq
4 (2 a_{1}+2 n a_{3}+a_{5})-c_{1}+(1-n)c_{2}=0 \label{a1}
\eeq
\beq
4 (2 a_{2}+2  a_{4}+n a_{5})+c_{1}+(1-n)c_{3}=0
\eeq
\beq
-2 b_{1}+b_{2}+(1-n) b_{3}+2(c_{1} +n c_{2}+c_{3})=0 \label{a3}
\eeq
which in turn, by virtue of our Theorem, imply that our starting action\footnote{The same result was also used earlier in \cite{jimenez2018teleparallel} in order to obtain a projective invariant torsion scalar.} is projective invariant for this parameter choice. Note that the above three conditions on the parameters are in perfect agreement with the constraints for a projective invariant Theory obtained in \cite{iosifidis2019scale} and here we obtained them almost without effort. Note also that the fourth constraint found there is redundant and can be shown to be a linear combination of $(\ref{a1})-(\ref{a3})$. Going one step further we can immediately obtain the parameter choice for which the above action is invariant under connection transformations of the second kind. By simply demanding a vanishing $2^{nd}$ contraction ($\Psi_{(2)}^{\mu}=0$) we get 
\beq
2 a_{2}+\frac{c_{1}}{2}+4 a_{3}+(n+1)a_{5}+\frac{(n-1)}{2}c_{2}=0 \label{a4}
\eeq  
\beq
4 a_{1}+2 a_{2}-\frac{c_{1}}{2}+2 ( a_{4}+a_{5})+\frac{(n-1)}{2}c_{3}=0
\eeq
\beq
2 b_{1}-b_{2}-c_{1}+2 c_{2}+(n+1)c_{3}+(n-1)b_{3}=0 \label{a6}
\eeq
Similarly, by demanding a vanishing $3^{rd}$ contraction ($\Psi_{(3)}^{\lambda}=0$) we  obtain the parameter space for the Theory that is invariant under transformations of the $3^{rd}$ $kind$ , which reads
\beq
2 a_{2}+4 a_{3}+(n+1)a_{5}=0 \label{a7}
\eeq
\beq
2 a_{1}+a_{2}+a_{5}+(n+1) a_{4}=0
\eeq
\beq
-c_{1}+2 c_{2}+(n+1)c_{3}=0 \label{a9}
\eeq
Notice that there is no restriction on the  $b_{i}'s$. This is of course to be expected since the $b_{i}'s$ multiply the purely torsional terms and these, and torsion being antisymmetric in its first two indices is insensitive to transformations of the $3^{rd}$ $kind$. From the above exposure, we come to the following conclusions:
\begin{enumerate}
	\item If the parameters of the action ($\ref{genact}$) satisfy the set of equations $(\ref{a1})$-$(\ref{a3})$ then the action is invariant under projective transformations (or else transformations of the $1^{st}$ kind) of the affine connection and the $1^{st}$ contraction of its $\Gamma$-variation vanishes identically.
	\item If the parameters of the action ($\ref{genact}$) satisfy the set of equations $(\ref{a4})$-$(\ref{a6})$ then the action is invariant under  transformations of the $2^{nd}$ kind of the affine connection and the $2^{nd}$ contraction of its $\Gamma$-variation vanishes identically.
	\item If the parameters of the action ($\ref{genact}$) satisfy the set of equations $(\ref{a7})$-$(\ref{a9})$ then the action is invariant under  transformations of the $3^{nd}$ kind of the affine connection and the $3^{nd}$ contraction of its $\Gamma$-variation vanishes identically.
\end{enumerate}
In addition if we combine all the above restrictions on the parameters we are lead to the following conclusion:
 
\begin{itemize}
\item The general quadratic action $(\ref{genact})$  is invariant under general vectorial transformations of the connection of the form ($\ref{vecgen}$) iff the parameters satisfy the set of equations  $(\ref{a1})$-$(\ref{a3})$, $(\ref{a4})$-$(\ref{a6})$ and $(\ref{a7})$-$(\ref{a9})$.
\end{itemize}

\subsection{Constraints on The Matter Sector}
The above transformations of the connection and the subsequent invariances of the Gravity sectors are translated to certain constraints on the hypermomentum tensor, once applied to the matter part of the Theory. Indeed, let us consider a general MAG Lagrangian\footnote{Of course, the Gravitational part will have an explicit dependence on the curvature, torsion and non-metricity tensors. All of the aforementioned tensors, however, are derivable from the metric, the connection and their derivatives. So in the end the basic variables are the metric, the connection and their derivatives.}
\beq
S[g,\Gamma,\phi]=\frac{1}{2 \kappa}S_{G}[g,\Gamma]+S_{M}[g,\Gamma,\phi] \label{SG}
\eeq
where $S_{G}$ denotes the Gravitational sector (Geometry) and $S_{M}$ represents the matter fields. The equations of motion then follow by varying independently the total action with respect to the metric and the connection, and one finds
\beq
Z_{\mu\nu}=\kappa T_{\mu\nu}
\eeq
\beq
P_{\lambda}^{\;\;\mu\nu}=\kappa \Delta_{\lambda}^{\;\;\;\mu\nu}
\eeq
where
\beq
Z_{\mu\nu}:=\frac{1}{\sqrt{-g}}\frac{\delta S_{G}}{\delta g^{\mu\nu}} \;\;,\;\;\; P_{\lambda}^{\;\;\mu\nu}:=\frac{1}{\sqrt{-g}}\frac{\delta S_{G}}{\delta \Gamma^{\lambda}_{\;\;\;\mu\nu}} \nonumber
\eeq
for the metric and connection variations respectively. In the above $T_{\mu\nu}$ is the usual metrical energy momentum tensor
\beq
T_{\mu\nu}:=-\frac{2}{\sqrt{-g}}\frac{\delta S_{M}}{\delta g^{\mu\nu}}
\eeq
and $\Delta_{\lambda}^{\;\;\;\mu\nu}$ is the so-called hypermomentum tensor\cite{hehl1976hypermomentum}
\beq
\Delta_{\lambda}^{\;\;\;\mu\nu}:=-\frac{2}{\sqrt{-g}}\frac{\delta S_{M}}{\delta \Gamma^{\lambda}_{\;\;\;\mu\nu}}
\eeq
which  encompasses the microscopic characteristics of matter such as spin, shear and dilation\cite{hehl1995metric}. The couplet ($T_{\mu\nu}$ , $\Delta_{\lambda}^{\;\;\mu\nu}$) represent the  sources for MAG. From the above it is clear that if the Gravity sector is invariant under certain transformations of the affine connection (of the form $(\ref{affcongen})$) then this invariance puts constraints on the form of the hypermomentum. In other words, it constrains the types of matter fields that can enter the spacetime. Perhaps the most common example of this, is the projective invariance of the Einstein-Hilbert action\footnote{That is $R$ is invariant under projective transformations of the connection.} which imposes a vanishing dilation current
\beq
\Delta_{\lambda}^{\;\; \lambda \mu}=0 \label{constD}
\eeq
These constraints may be unphysical in a general setting, however as noted in \cite{jimenez2018teleparallel} both fermionic and bosonic matter is projective invariant. Therefore, whether  constraints of the form $(\ref{constD})$ are unphysical or not is not such an easy question to answer and goes beyond the scope of this letter.

 Following the Corollaries we obtained in the previous section, when applied to the matter sector of a MAG Theory, they constrain the material sources. Indeed, setting
\beq
\Delta_{(1)}^{\nu}:=\Delta_{\lambda}^{\;\;\lambda\nu}\;, \; \Delta_{(2)}^{\mu}:=\Delta_{\lambda}^{\;\;\mu\lambda}\;, \; \Delta_{(3)\lambda}:=\Delta_{\lambda}^{\;\;\mu\nu}g_{\mu\nu}
\eeq
and applying the Corollaries of the previous section on the Theory $(\ref{SG})$ we are lead to the following conclusions.

\begin{prop}
	If the gravity sector of ($\ref{SG}$) is invariant under connection transformations of the $i-th$ $kind$ ($i=1,2,3$) then this invariance demands a vanishing $i-th$ contraction of the hypermomentum
	\beq
	\Delta_{(i)}^{\mu}=0
	\eeq
\end{prop}

\begin{prop}
		If the gravity sector of ($\ref{SG}$) is invariant under general vectorial transformations of the connection of the form ($\ref{vecgen}$), then the hypermomentum tensor has all of its contractions vanishing, namely
	\beq
	\Delta_{(i)}^{\mu}=0\;, \forall i=1,2,3
	\eeq

\end{prop}

\begin{prop}
	If the gravity sector of ($\ref{SG}$) is invariant under constrained vectorial transformations of the form ($\ref{veccons}$), then the traces of the hypermomentum of the Theory satisfy the constraint
	\beq
	\sum_{i=1}^{3}\lambda_{i}\Delta_{(i)}^{\mu}=0
	\eeq
	That is, the three traces of the hypermomentum are  linearly dependent.
\end{prop}
Note that special cases of the above give the results of $Proposition$ $1$. In addition for the parameter choice $\lambda_{1}=\lambda_{2}=-\lambda_{3}$ one obtains invariance under a  Weyl transformation of the connection. Then if the Theory is invariant under Weyl transformations, the hypermomentum traces satisfy
\beq
\Delta_{(1)}^{\mu}+\Delta_{(2)}^{\mu}-\Delta_{(3)}^{\mu}=0
\eeq
We shall call such type of matter, a $Weyl$-matter. It would be interesting to study Theories with such matter fields and also see if these kind of constraints are purely unphysical or are indeed obeyed by some certain physical systems. Certainly, this issue deserves further investigation. If it turns out that all of these constraints are unphysical, our result could be applied to rule out Gravitational sectors of Metric-Affine Theories that respect certain symmetries of the connection. We should note that much milder constraints are put on the hypermomentum traces if the Gravitational action is invariant under special vectorial transformations of the connection (i.e the one forms of the transformation are exact, for instance $\xi_{\mu}=\partial_{\mu}\lambda$). Indeed, if the action $(\ref{SG})$ is invariant under what Einstein called, $\lambda$ (special projective) transformations
\beq
\Gamma^{\lambda}_{\;\;\mu\nu}\longrightarrow \Gamma^{\lambda}_{\;\;\mu\nu}+\delta^{\lambda}_{\mu}\partial_{\nu}\lambda
\eeq
Then, up to surface terms, we obtain
\beq
\partial_{\mu}(\sqrt{-g}\Delta^{\mu}_{(1)})=0
\eeq
which has the form of a conservation law (see also \cite{ponomariov1982generalized,iosifidis2019scale}). Obviously, for special vectorial transformations all of the above propositions change accordingly, and the constraints are imposed not on the traces of the hypermomentum, but rather on their divergences. For instance, if the Gravitational action is invariant under special constrained vectorial transformations of the form
\beq
\Gamma^{\lambda}_{\;\;\mu\nu}\longrightarrow \Gamma^{\lambda}_{\;\;\mu\nu}+c_{1}\delta^{\lambda}_{\mu}\partial_{\nu}\lambda+c_{2}\delta^{\lambda}_{\nu}\partial_{\mu}\lambda+c_{3}g_{\mu\nu}\partial^{\lambda}\lambda 
\eeq 
Then the  divergences of the traces of the hypermomentum tensor, obey
	\beq
\sum_{i=1}^{3}c_{i}\partial_{\mu}(\sqrt{-g}\Delta_{(i)}^{\mu})=0
\eeq
Note that the above milder constraints have a much more pleasing appeal then their stronger counterparts discussed above. Again, it would be very interesting to investigate the physical consequences of these constraints and comment upon their viability. This would be a topic to be pursued elsewhere.

\section{Conclusions}
Non-Riemannian Geometries form the geometric arena in which Metric-Affine Theories of Gravitation live. In this extended geometry, the linear connection and the metric are independent fields and the space is also endowed with torsion and non-metricity on top of curvature. In this context the linear connection and the metric can be transformed independently and is important to have a tool that allows one to explore the invariance of the action (or a functional in general) under a certain group of transformations. 

In this paper we presented and proved a Theorem that allows one to obtain invariant functionals (or specifically invariant  actions if we are dealing with Gravity) under certain transformations of the affine connection. More specifically, we showed  that if a functional is invariant under a certain transformation of the affine connection, its $\Gamma$-variation produces a tensor of a particular symmetry. In addition, the nice feature of our Theorem is that it also works backwards, that is if the tensor obtained by the $\Gamma$-variation of the functional has a certain symmetry then the initial functional is invariant under a certain transformation of the affine connection. An immediate consequence of the Theorem is that if a functional is invariant under connection transformations of the $i-$th kind\footnote{See definitions at section \textbf{IV}.} ($i=1,2,3$) then the  $i-th$ contraction of its $\Gamma-$variation vanishes and conversely. It also follows that if all of the contractions of the $\Gamma$-variation of a functional vanish, the functional is invariant under the general vectorial transformation ($\ref{vecgen}$) of the connection. Additionally, if  the three traces of of the $\Gamma$-variation of the functional are linearly dependent then the later is invariant under a constraint vectorial transformation of the connection of the form $(\ref{veccons})$.

Even though up to this point all the above seem to be only of mathematical interest,we should note that the results find a natural place of application in Metric-Affine Gravity.
Indeed, as we showed, when the functional is taken to represent a Gravitational (or matter) action the Theorem can be applied to give us information about the symmetries of the action. As an example of the above we reproduced and also extended the results found in \cite{iosifidis2019scale} regarding the most general quadratic projective invariant Theory. In addition, we found the most general quadratic (in torsion and non-metricity) whose $\Gamma$-variation produces a completely traceless tensor in all of its contractions. We then showed how the invariance of a MAG Gravitational action under certain transformations of the connection is related to constraints imposed on the hypermomentum tensor. Considering vectorial  (and also special) connection transformations we then obtained the constraints imposed on the hypermomentum.

\section{Acknowledgments}
I would like to thank Jose Beltran Jimenez for useful discussions and comments and also Tomi Sebastian
Koivisto for some remarks. I am also grateful to the anonymous referees for their helpful critique and suggestions. This research is co-financed by Greece and the European Union (European Social Fund- ESF) through the
Operational Programme «Human Resources Development, Education and Lifelong Learning» in the context
of the project “Reinforcement of Postdoctoral Researchers - 2
nd Cycle” (MIS-5033021), implemented by the
State Scholarships Foundation (IKY).

\bibliographystyle{unsrt}
\bibliography{ref}

\end{document}